\documentclass[conference]{IEEEtran}
\ifCLASSINFOpdf
\else
\fi
\usepackage[T1]{fontenc}
\usepackage{graphicx}
\usepackage{subfigure}
\usepackage{amsmath}
\usepackage{amsthm}
\usepackage{amsfonts}
\usepackage{dsfont}
\usepackage{amssymb}
\newtheorem{Theorem}{Theorem}

\newtheorem{Lemma}{Lemma}

\usepackage{balance}
\usepackage{paralist}
\newcommand{\argmax}{\operatornamewithlimits{arg\,max}}
\newcommand{\argmin}{\operatornamewithlimits{arg\,min}}
\usepackage{algorithm}
\usepackage{algorithmic}

\begin{document}

\title{Dynamic Mobile Edge Caching with \\Location Differentiation}  
\author{\IEEEauthorblockN{Peng Yang\IEEEauthorrefmark{1},
Ning Zhang\IEEEauthorrefmark{2},
Shan Zhang\IEEEauthorrefmark{2}, 
Li Yu\IEEEauthorrefmark{1},
Junshan Zhang\IEEEauthorrefmark{3}, and
Xuemin (Sherman) Shen\IEEEauthorrefmark{2}}
\IEEEauthorblockA{\IEEEauthorrefmark{1}School of Electronic Information and Communications, Huazhong University of Science and Technology, Wuhan, China}
\IEEEauthorblockA{\IEEEauthorrefmark{2}Department of Electrical and Computer Engineering, University of Waterloo, Waterloo, Canada}
\IEEEauthorblockA{\IEEEauthorrefmark{3}School of Electrical, Computer and Energy Engineering, Arizona State University, Tempe, Arizona, USA \\Email: \IEEEauthorrefmark{1}\{yangpeng, hustlyu\}@hust.edu.cn, \IEEEauthorrefmark{2}\{n35zhang, s372zhan, sshen\}@uwaterloo.ca, \IEEEauthorrefmark{3}junshan.zhang@asu.edu}}

\maketitle

\begin{abstract}
Mobile edge caching enables content delivery directly within the radio access network, which  effectively alleviates the backhaul burden and reduces round-trip latency. To fully exploit the edge resources, the most popular contents should be identified and cached. Observing that content popularity varies greatly at different locations, to maximize local hit rate, this paper proposes an online learning algorithm that dynamically predicts content hit rate, and makes location-differentiated caching decisions. Specifically, a linear model is used to estimate the future hit rate. Considering the variations in user demand, a perturbation is added to the estimation to account for uncertainty. The proposed learning algorithm requires no training phase, and hence is adaptive to the time-varying content popularity profile. Theoretical analysis indicates that the proposed algorithm asymptotically approaches the optimal policy in the long term. Extensive simulations based on real world traces show that, the proposed algorithm achieves higher hit rate and better adaptiveness to content popularity fluctuation, compared with other schemes.
\end{abstract}

\section{Introduction}
The soaring mobile traffic has put high pressure on the paradigm of Cloud-based service provisioning, because moving a large volume of data into and out of the Cloud wirelessly consumes substantial spectrum resources, and meanwhile may incur large latency. \emph{Mobile Edge Computing} (MEC) emerges as a new paradigm to alleviate the capacity concern of mobile access networks \cite{ben17}. Residing on the network edge, MEC makes storage and computing resources available to mobile users through one-hop wireless connections, facilitating a number of mobile services, such as local content caching, augmented reality, and cognitive assistance \cite{ETSI}.

Among other services, content caching on the edge gains increased attention \cite{fran}-\cite{icc16}. In particular, with the prevalence of social media, contents, such as high-resolution videos, are spreading among mobile users in viral fashion, putting tremendous pressure on the network backhaul \cite{tsp,tit13}. It is forecast that, caching contents on network edge can reduce up to $35\%$ traffic demand on the backhaul \cite{ETSI}. Unfortunately, compared with the ever-increasing content volume, the storage on the edge node (EN) is always limited. It is impossible to cache all the contents locally. Hence, it becomes crucial to identify the optimal set of contents that maximizes the cache utilization.

Content popularity is an effective measure for making caching decisions, based on which influential contents can be identified and cached proactively \cite{jian}. Yet, content popularity is unknown \emph{a priori}, it is hard to directly select and cache the most popular ones. 
Blasco \emph{et al.} developed an online algorithm to learn the content popularity profile \cite{icc14}. Specifically, it predicts future content hit rate based on the number of instantaneous requests of cached contents, and then makes new caching decisions correspondingly. In practice, however, the popularity profile of content is not only unknown, but also varying since user's interests are constantly changing \cite{acmcom}, and meanwhile new contents are being created. To maximize cache utilization under the condition of varying and unknown popularity profile, M\"uller \emph{et al.} proposed a context-aware caching algorithm based on user information, which includes users' ages, genders or their preferences \cite{icc16}. 
However, relying on user information for context differentiation is risky since such information is extremely sensitive and often unavailable.
Alternatively, exploiting location features for context differentiation is a feasible approach. Generally, locations can be classified according to their social functions, such as residential area and business district. Users in different areas have diverse interests \cite{acmcom}. Hence, locational statistics, including the number of users and regional content preferences, can be used for context differentiation. Based on which fine-grained caching decisions can be made to improve the cache utilization.

In this paper, we investigate the problem of mobile edge caching with location differentiation. Since the ability of identifying popular contents is crucial, this problem is challenging in the following ways. Firstly, the future hit rate of a content at a certain location is unknown ahead. Secondly, though the popularity diversity among different locations is evident, there is no established model of how location features affect content hit rate. Thirdly, content hit rate is varying continuously, so the caching strategy should dynamically adjust to the changes. To address those issues, we propose a novel learning algorithm that estimates future content hit rate based on a linear prediction model. This model incorporates content feature and location characteristics, with well-balanced stability and estimation accuracy. Considering the impact of random noise, a subtly designed perturbation is added to the prediction of the linear model to account for uncertainty. Theoretical analysis indicates that the proposed algorithm achieves sublinear regret, i.e., it asymptotically approaches the optimal strategy in the long term. Extensive simulations based on real world traces show that the proposed caching algorithm achieves better accuracy on hit rate prediction, and meanwhile adapts steadily to the popularity dynamics.

The remainder of the paper is organized as follows. Section \ref{systemmodel} describes the system model and the caching problem formulation. Section \ref{cachingalgorithm} presents the proposed location differentiated caching algorithm, followed by the theoretical regret analysis in Section \ref{regretanalysis}. Simulation results are shown in Section \ref{simulation} and concluding remarks are given in Section \ref{conclusion}.

\section{System Model and Problem Formulation}\label{systemmodel}
\subsection{Network Model}
\begin{figure}[t]
\centering
\includegraphics[width=0.325\textwidth]{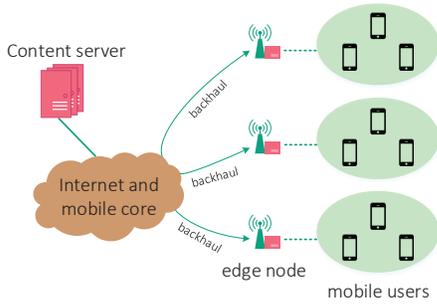}
\caption{Network model of mobile edge caching.}
\label{cachingmodel}
\end{figure}

\begin{figure}[t]
\centering
\includegraphics[width=0.35\textwidth]{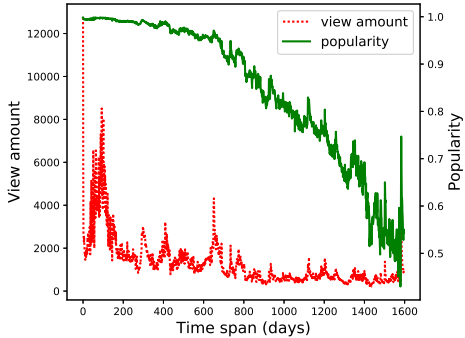}
\caption{The daily view amount and popularity curves of a YouTube video since uploaded. Note that the popularity score is calculate based on the statistics of a set of randomly crawled videos.}
\label{Youtubestatistics}
\end{figure}

Capacity-augmented base stations, WiFi access points and other devices with excess capacity can be exploited for EN deployment \cite{ben17}. In this paper, the storage resources on ENs are harnessed for content caching services. Specifically, consider the mobile edge caching model illustrated in Fig. \ref{cachingmodel}, a set of ENs $\mathcal N = \{1,2,\dotsc,N\}$ is deployed with separated backhaul links connecting to the mobile core network. Each EN $n$ is associated with a distinct location and has different characteristics in terms of content requests compared with others. Online contents are dynamically pushed to ENs so that user's content requests can be processed with reduced latency. Each EN serves a disjoint set of mobile users.  

\subsection{Content Popularity with Location Differentiation}
A simple yet effective caching strategy is pushing the most popular contents to the network edge. Hence, local content hit rate is maximized and user requests are served with reduced latency and improved quality of experience. Extensive works have been done on the popularity of contents, especially video files \cite{tsp,acmcom,imc}. According to the statistics we crawled from YouTube, as illustrated in Fig. \ref{Youtubestatistics}, the popularity profile of a video file varies in two-fold. On one hand, the daily view amount is varying. On the other hand, as other videos' daily view amounts are also varying and new videos are being uploaded, the popularity score of a video file is constantly fluctuating. Moreover, location-related characteristics also affect the content popularity. As a result, general caching strategies that based on fixed popularity profile are not optimal in practice. 

Consider a set of files $\mathcal F = \{1,2,\dotsc,F\}$ that can be cached at ENs, and let $c<F$ be the caching size of each EN. We assume that all contents are of equal size\footnote{In case contents are of different sizes, they can be split into smaller ones of equal size. For example, the widely used DASH (Dynamic Adaptive Streaming over HTTP) protocol breaks content into small segments before transmission.} and the size is normalized to 1, i.e., each EN can cache up to $c$ contents. We focus on those files' popularity dynamics in a sequence of time slots $\mathcal T = \{1,2,\dotsc,T\}$. Let $\boldsymbol{x}_{f,n,t} \in \mathbb R^d$ be a $d$-dimensional feature vector of file $f$ associated with EN $n$ observed before time slot $t$. For a certain EN $n$, the hit rate\footnote{We define \emph{hit rate} as the number of content requests rather than a ratio.} of file $f$ during time slot $t$, denoted by $d_{f,n,t}$, is statistically linear with respect to its feature vector $\boldsymbol{x}_{f,n,t}$, i.e.,
\begin{equation}\label{linearprediction}
\mathbb E[d_{f,n,t}|\boldsymbol{x}_{f,n,t}] = \boldsymbol{x}_{f,n,t}^\top \boldsymbol{\theta}_{n}^\ast,
\end{equation}
where $\boldsymbol{\theta}_{n}^\ast \in \mathbb R^d$ is the unknown true parameter vector associated with EN $n$. The vector $\boldsymbol{x}_{f,n,t}$ between file $f$ and EN $n$ may contain feature information like the frequency of file $f$ being cached at EN $n$ up to time $t$, the hit rate of $f$ at EN $n$ during the last $5$ time slots, $30$ time slots etc. The parameter vector $\boldsymbol{\theta}_{n}^\ast$ represents the specific location characteristics at EN $n$ and determines the particular combination of different features on the expected hit rate. In this way, a certain content is expected to have different hit rates at different ENs. This linear prediction model is widely used in other areas like signal processing and financial engineering \cite{feng,chu}. It provides a method to predict content hit rate, which is essential to proactive content caching.

\subsection{Problem Formulation}
As indicated by Fig. \ref{Youtubestatistics}, the popularity of a content is varying constantly. Hence, we intend to perform dynamic content caching that constantly updates the files on ENs to achieve higher long-term hit rate. To this end, contents with higher popularity at different locations should be proactively identified and cached respectively, and meanwhile the less popular ones should be evicted. Let $\mathcal F_{n,t}$ denote the set of contents cached at EN $n$ during time slot $t$, then the dynamic caching problem can be formulated as the following long-term hit rate maximization (LHRM) problem:
\begin{equation}
\begin{array}{rl}
\mbox{(\underline{LHRM}):}\max & \!\sum_{t\in\mathcal T} \sum_{n \in \mathcal N} \sum_{f \in \mathcal F_{n,t}} d_{f,n,t} \\
\textrm{Subject to:} & \!|\mathcal F_{n,t}| \leq c, \; \forall n\in\mathcal N, \; t \in \mathcal T.\label{LHRM}
\end{array}
\end{equation}
However, the popularity profile, i.e., the hit rate of contents at each EN, is unknown \emph{a priori}. Hence the decision variables $\mathcal F_{n,t}$ in above optimization problem is intractable directly. For convenience, denote the optimal caching strategy $\mathcal F_{n,t}^\ast$ for EN $n$ at time $t$. We have that
\begin{equation}
\mathcal F_{n,t}^\ast = \argmax_{|\mathcal F_{n,t}| \leq c}\sum_{f \in \mathcal F_{n,t}}d_{f,n,t},\; \forall n \in \mathcal N,\; t \in \mathcal T.
\end{equation}
Define the long-term \emph{regret} of a solution respect to the optimal caching strategy as 
\begin{equation}\label{regret}
R(T) \triangleq \mathbb E\left[ \sum_{t\in\mathcal T} \sum_{n \in \mathcal N} \left(\sum_{f \in \mathcal F_{n,t}^\ast} d_{f,n,t} - \sum_{f \in \mathcal F_{n,t}} d_{f,n,t}\right) \right].
\end{equation}
Then, the LHRM problem can be reformulate as a long-term regret minimization (LRM) problem: 
\begin{equation}
\begin{array}{rl}
\mbox{(\underline{LRM}):}\min & R(T)\\
\textrm{Subject to:} & \!|\mathcal F_{n,t}| \leq c, \; \forall n\in\mathcal N, \; t \in \mathcal T.\label{LRM}
\end{array}
\end{equation}
In the long term, if an algorithm can continuously identify the optimal set $\mathcal F_{n,t}^\ast$ and cache those files at the corresponding ENs, the algorithm achieves \emph{zero-regret}. Since $\mathcal F_{n,t}^\ast$ is unknown \emph{a priori}, our goal is to develop a caching algorithm that makes good estimation on future content hit rate, and hence better identifies the popular ones. To this end, we propose an online learning algorithm that dynamically adjusts the estimation of location parameter vectors. The estimation error is carefully bounded so that the proposed algorithm asymptotically approaches the optimal caching policy.

\section{Location Differentiated Content \\Caching Algorithm}\label{cachingalgorithm}
To better characterize different locations and make accurate prediction on future content hit rate, we resort to the linear model given by Eq. (\ref{linearprediction}). It can be interpreted that, at time slot $t$, given the feature vector $\boldsymbol{x}_{f,n,t}$, the hit rate of file $f$ at EN $n$ is predicted to be the linear combination of the features, which gives a feasible way to predict the content hit rate. However, the parameter vector for combination is unknown \emph{a priori}. Therefore, a good estimation of the true parameter vector $\boldsymbol{\theta}_n^\ast$ will lead to accurate prediction of the content hit rate. 

In this section, we first present a regression-based method to estimate the parameter vector of each EN. Then we present an online caching algorithm that predicts future content hit rate based on the continually evolving parameter vector.

\subsection{Predicting Content Hit Rate}
Consider a specific file $f$ and EN $n$, the estimation of parameter vector $\boldsymbol{\theta}_n^\ast$ can be performed in an online fashion based on file $f$'s historical data. Let $\boldsymbol{\Phi}_{f,n}\in \mathbb R^{m \times d}$ be the historical feature vectors of file $f$, where $m$ is the frequency of file $f$ being cached at EN $n$ up to time slot $t$, and the $m$-th row of $\boldsymbol{\Phi}_{f,n}$ is the corresponding feature vector $\boldsymbol{x}_{f,n,m}$. Denote $\boldsymbol{y}_{f,n}\in\mathbb R^m$ the $m$-time empirical hit rate of file $f$ at EN $n$. By applying the standard ordinary least square linear regression, i.e., $\boldsymbol{\theta}_n^\ast = {\argmin}_{\boldsymbol{\theta}_n} ||\boldsymbol{y}_{f,n} - \boldsymbol{\Phi}_{f,n}\boldsymbol{\theta}_n||$, the closed-form estimation of $\boldsymbol{\theta}_n^\ast$ can be derived as $(\boldsymbol{\Phi}_{f,n}^\top \boldsymbol{\Phi}_{f,n})^{-1}\boldsymbol{\Phi}_{f,n}^\top \boldsymbol{y}_{f,n}$. However, considering that feature vectors may be correlated, and hence the matrix $\boldsymbol{\Phi}_{f,n}^\top \boldsymbol{\Phi}_{f,n}$ could be singular, rendering the estimation of $\boldsymbol{\theta}_n^\ast$ fluctuate significantly. Instead of the unbiased estimation made by ordinary least square linear regression, ridge regression makes biased estimation by adding a control parameter that restricts the magnitude of the parameter vector, which helps to improve the estimation stability. By ridge regression, the estimation of $\boldsymbol{\theta}_n^\ast$ can be explicitly given as
\begin{equation}\label{ridge}
\tilde{\boldsymbol{\theta}}_n = (\boldsymbol{\Phi}_{f,n}^\top \boldsymbol{\Phi}_{f,n} + \lambda \boldsymbol{I}_d)^{-1} \boldsymbol{\Phi}_{f,n}^\top \boldsymbol{y}_{f,n},
\end{equation}
where $\boldsymbol{I}_d \in \mathbb R^{d \times d}$ is the identity matrix and $\lambda > 0$ is the control parameter that guarantees the stability of estimation. The accuracy of estimation depends on the amount of data and $\lambda$. For convenience, let $\boldsymbol{V}_{f,n} = \boldsymbol{\Phi}_{f,n}^\top \boldsymbol{\Phi}_{f,n} + \lambda \boldsymbol{I}_d$ for all $f \in \mathcal F$ and $n \in \mathcal N$. The following lemma, which is slightly manipulated from \cite{chu}, gives an upper bound on the estimation error of ridge regression.

\begin{Lemma}\label{estimationerror}
If $||\boldsymbol{\theta}_{n}^\ast|| \leq \zeta$ for all $n \in \mathcal N$, where $||\cdot||$ denotes the Euclidean norm. Then, $\forall \delta >0$, the estimation error of ridge regression can be upper bounded as 
\begin{equation}\label{errorbound}
|\boldsymbol{x}_{f,n}^\top \tilde{\boldsymbol{\theta}}_n - \boldsymbol{x}_{f,n}^\top \boldsymbol{\theta}_n^\ast| \leq (\delta + \zeta\lambda) \sqrt{\boldsymbol{x}_{f,n}^\top \boldsymbol{V}_{f,n}^{-1}\boldsymbol{x}_{f,n}}
\end{equation}
with probability at least $1-2e^{-2\delta^2}$.
\end{Lemma}
Please refer to Appendix for the proof. The upper bound of estimation error provided in Lemma \ref{estimationerror} can be interpreted that, the true hit rate falls into the confidence interval around the estimation with high probability. Based on Lemma \ref{estimationerror}, we propose a dynamic content hit rate prediction and caching algorithm.

\subsection{Caching Algorithm}
\begin{algorithm}[t] 
\caption{Location Differentiated Edge Caching Algorithm}
\begin{algorithmic}[1] \label{LinUCB}
\REQUIRE $\lambda >0$.
\ENSURE Set of files to be cached in each EN.
\STATE Initialization: Cache files in all ENs and get the initial feature vectors $\boldsymbol{x}_{f,n,0}$ of all file-EN pairs.
\STATE $\boldsymbol{V}_{n} \gets \lambda \boldsymbol{I}_d$, $\boldsymbol{h}_{n} \gets \boldsymbol{0}_d,\;\forall n\in\mathcal N$
\FOR {$t = 1,2,\dotsc,T$}
\FOR {each EN $n \in \mathcal N$}
\STATE $\tilde{\boldsymbol\theta}_{n,t} \gets \boldsymbol{V}_{n}^{-1} \boldsymbol{h}_{n}$
\FOR {each file $f \in \mathcal F$}
\STATE Obtain feature vectors $\boldsymbol{x}_{f,n,t}$
\STATE $\tilde d_{f,n,t} \gets \boldsymbol{x}_{f,n,t}^\top\tilde{\boldsymbol\theta}_{n,t}, \; \hat 		d_{f,n,t} \gets \tilde d_{f,n,t} + p_{f,n,t}$
\ENDFOR
\STATE $\mathcal F_{n,t} = \argmax_{\mathcal F_n \subseteq \mathcal F,\;|\mathcal F_n| \leq c} 			\sum_{f\in\mathcal F_n} \hat d_{f,n,t}$
\STATE Cache all the files in set $\mathcal F_{n,t}$ on EN $n$
\STATE Observe the empirical hit rate $d_{f,n,t}$ of cached files 
\STATE Update $\boldsymbol{V}_{n}$ and $\boldsymbol{h}_{n}$ based on $\boldsymbol{x}_{f,n,t}$ and 		$d_{f,n,t}$ of all cached files: \\
		$\boldsymbol{V}_{n} \gets \boldsymbol{V}_{n} + \boldsymbol{x}_{f,n,t}\boldsymbol{x}		_{f,n,t}^\top$ \\
		$\boldsymbol{h}_{n} \gets \boldsymbol{h}_{n} + \boldsymbol{x}_{f,n,t}d_{f,n,t}$
\ENDFOR
\ENDFOR
\end{algorithmic}
\end{algorithm}
The location differentiated EN caching algorithm is sketched in Algorithm \ref{LinUCB}. During each time slot, the algorithm first updates the estimated parameter vector $\tilde{\boldsymbol\theta}_{n,t}$. With the accumulation of historical data, $\tilde{\boldsymbol\theta}_{n,t}$ will finally converge to the true parameter vector $\boldsymbol{\theta}_{n}^\ast$. Then, based on the instantaneous content feature vector $\boldsymbol{x}_{f,n,t}$, the predicted hit rate $\tilde d_{f,n,t}$ is obtained according to the linear model. Furthermore, a perturbation term $p_{f,n,t}$ is added to the linear prediction, where
\begin{equation}\label{perturbation}
p_{f,n,t} = \alpha_t \sqrt{\boldsymbol{x}_{f,n,t}^\top \boldsymbol{V}_{f,n}^{-1} \boldsymbol{x}_{f,n,t}},
\end{equation}
and $\alpha_t = \big[\ln (tF^{\frac{1}{2}})\big]^{\frac{1}{2}} + \zeta\lambda$. The rationale of the perturbation is that Eq. (\ref{linearprediction}) only gives an expectation value of the hit rate while omitting the potential random noises. The perturbation specified by Eq. (\ref{perturbation}) is inline with Lemma \ref{estimationerror} and can be regarded as the optimism in face of uncertainty, or equivalently, the upper confidence of the predicted hit rate. With the above setting of $\alpha_t$, we have $\delta =  \big[\ln (tF^{\frac{1}{2}})\big]^{\frac{1}{2}}$ in Lemma \ref{estimationerror}. Hence, as $t$ increases, the upper confidence bound holds with high probability (at least $1-2F^{-1}t^{-2}$). Based on the upper confidence $\hat d_{f,n,t}$ of predicted hit rate, a set of file $\mathcal F_{n,t}$ that is predicted to maximize the content hit rate at EN $n$ is cached respectively. Afterwards, the empirical hit rate information of all cached files is recorded, which is used to update the database for subsequent estimation and prediction. Note that, a file may be simultaneously cached in multiple ENs.

\section{Regret Analysis}\label{regretanalysis}
The long-term hit rate of the proposed algorithm highly depends on the accuracy of prediction. This section gives a theoretical upper bound on the regret of long-term hit rate of the proposed algorithm. 

In mobile edge caching, let $c$ be the caching size of each EN, and $F$ be the size of ground file set. Suppose content hit rate satisfies the linear model, and feature vectors are bounded by  $||\boldsymbol{x}_{f,n,t}|| \leq \eta$ for all $f \in \mathcal F$, $n \in \mathcal N$ and $t\in\mathcal T$, where $||\cdot||$ denotes the Euclidean norm. We have the following theorem.

\begin{Theorem}\label{theorem}
Mobile edge caching Algorithm \ref{LinUCB} achieves sublinear long-term regret. Specifically, the long-term regret $R(T)$ is at most of order $O(cN\sqrt{dT (\ln T) \ln(\lambda + T\eta^2/d)})$. 
\end{Theorem}

\begin{proof}
The total regret depends on the algorithm's accuracy of estimation on content hit rate, which is elaborated in Lemma \ref{estimationerror}. According to this lemma, the true hit rate of file $f$ at EN $n$ lies in the confidence interval around the predicted hit rate
\begin{equation}
\mathcal I_{f,n,t} = [\boldsymbol{x}_{f,n,t}^\top \tilde{\boldsymbol{\theta}}_{n,t} - p_{f,n,t}, \; \boldsymbol{x}_{f,n,t}^\top \tilde{\boldsymbol{\theta}}_{n,t} + p_{f,n,t}]
\end{equation}
with high probability. 

Let $\mathcal X_{n,t} = \{\exists f \in\mathcal F: |d_{f,n,t} - \tilde d_{f,n,t}| \geq p_{f,n,t}\}$ be the event that there exists at least one file whose true hit rate lies outside its confidence interval. Let $\bar{\mathcal X}_{n,t}$ be the complementary event of $\mathcal X_{n,t}$, i.e., all files' true hit rates fall into their confidence interval. Let $r_{n,t}$ be the instant regret of a caching algorithm in EN $n$ at time slot $t$. According to Eq. (\ref{regret}), the total regret depends on the difference between the set of files chosen by the Algorithm and the optimum set, i.e., $\mathcal F_{n,t}$ and $\mathcal F_{n,t}^\ast$, thus
\begin{equation}\label{instantregret}
\begin{array}{rcl}
r_{n,t} &=& \sum_{f \in \mathcal F_{n,t}^\ast} d_{f,n,t} - \sum_{f \in \mathcal F_{n,t}} d_{f,n,t},
\end{array}
\end{equation}
and the long-term regret can be rewritten as
\begin{eqnarray}
R(T) &\!\!\!\!=\!\!\!\!& \sum_{t\in\mathcal T} \sum_{n \in \mathcal N} r_{n,t} \nonumber\\
&\!\!\!\!=\!\!\!\!& \sum_{t\in\mathcal T} \sum_{n \in \mathcal N} \mathds 1_{\{\mathcal X_{n,t}\}}r_{n,t} + \sum_{t\in\mathcal T} \sum_{n \in \mathcal N} \mathds 1_{\{\bar{\mathcal X}_{n,t}\}} r_{n,t},\label{dividedregret}
\end{eqnarray}
where $\mathds 1_{\{\mathcal X_{n,t}\}}$ is an indicator variable that equals to $1$ if event $\mathcal X_{n,t}$ happens and equals to $0$ otherwise. To bound the long-term regret, the two terms in Eq. (\ref{dividedregret}) are bounded respectively.

Firstly, consider the case when event $\mathcal X_{n,t}$ happens. With the setting of $\alpha_t$ in Eq. (\ref{perturbation}), for a file $f$ and EN $n$ at time $t$, we have $\mathbb P\{|d_{f,n,t} - \tilde d_{f,n,t}|\geq p_{f,n,t}\} \leq 2F^{-1}t^{-2}$. As a result, the frequency of event $\mathcal X_{n,t}$ happens in all ENs across the time span can be bounded as: 
\begin{eqnarray}
\sum_{t\in\mathcal T} \sum_{n \in \mathcal N} \mathds 1_{\{\mathcal X_{n,t}\}} &\!\!\!\leq\!\!\!& \sum_{t\in\mathcal T} \sum_{n \in \mathcal N} \sum_{f \in \mathcal F} \mathbb P\big\{|d_{f,n,t} - \tilde d_{f,n,t}|\geq p_{f,n,t}\big\} \nonumber\\
&\!\!\!\leq\!\!\!& \sum_{t\in\mathcal T} \sum_{n \in \mathcal N} \sum_{f \in \mathcal F} 2F^{-1}t^{-2} = 2N \sum_{t\in\mathcal T} t^{-2} \nonumber\\
&\!\!\!\leq\!\!\!& 2N \sum_{t=1}^\infty t^{-2} \leq \frac{\pi^2}{3}N.
\end{eqnarray}
Without loss of generality, let the content hit rate $d_{f,n,t}\leq \gamma,\;\forall f\in\mathcal F,\; n\in\mathcal N$ and $t\in\mathcal T$. According to Eq. (\ref{instantregret}), a coarse upper bound of the instant regret is $r_{n,t} \leq c \gamma$. Therefore, the first term of Eq. (\ref{dividedregret}) can be bound as 
\begin{equation}\label{firstbound}
\begin{array}{rcl}
\sum_{t\in\mathcal T} \sum_{n\in\mathcal N} \mathds 1_{\{\mathcal X_{n,t}\}} r_{n,t} &\leq& \pi^2 c\gamma N/3. \end{array}
\end{equation}

Then, consider the case when event $\bar{\mathcal X}_{n,t}$ happens, all files' true hit rates falls in to the confidence interval around their estimation $\tilde d_{f,n,t}$. Hence, $|d_{f,n,t} - \tilde d_{f,n,t}|\leq p_{f,n,t}, \; \forall f\in\mathcal F$. With $\hat d_{f,n,t} = \tilde d_{f,n,t} + p_{f,n,t}$, we have 
\begin{equation}\label{confidence}
0 \leq \hat d_{f,n,t} - d_{f,n,t} \leq 2p_{f,n,t}.
\end{equation} 
By Eq. (\ref{instantregret}) and (\ref{confidence}), when event $\bar{\mathcal X}_{n,t}$ happens, the instant regret $r_{n,t}$ can be bounded as
\begin{eqnarray}
r_{n,t}|_{\bar{\mathcal X}_{n,t}} &=& \sum_{f \in \mathcal F_{n,t}^\ast \setminus \mathcal F_{n,t}} d_{f,n,t} - \sum_{f \in \mathcal F_{n,t} \setminus \mathcal F_{n,t}^\ast} d_{f,n,t} \nonumber\\
&\leq& \sum_{f \in \mathcal F_{n,t}^\ast \setminus \mathcal F_{n,t}} \hat d_{f,n,t} - \sum_{f \in \mathcal F_{n,t} \setminus \mathcal F_{n,t}^\ast} d_{f,n,t} \nonumber\\
&\leq& \sum_{f \in \mathcal F_{n,t} \setminus \mathcal F_{n,t}^\ast} \Big(\hat d_{f,n,t} - d_{f,n,t}\Big) \label{setfunction}\\
&\leq& 2\sum_{f \in \mathcal F_{n,t} \setminus \mathcal F_{n,t}^\ast} p_{f,n,t}.
\end{eqnarray}
where inequality (\ref{setfunction}) is due to fact that since the algorithm selects files in $\mathcal F_{n,t} \setminus \mathcal F_{n,t}^\ast$ rather than $\mathcal F_{n,t}^\ast \setminus \mathcal F_{n,t}$, hence the collective upper confidence bound hit rate satisfies $\sum_{f \in \mathcal F_{n,t} \setminus \mathcal F_{n,t}^\ast} \hat d_{f,n,t} \geq \sum_{f \in \mathcal F_{n,t}^\ast \setminus \mathcal F_{n,t}} \hat d_{f,n,t} $. 
\begin{figure*}
\centering
\subfigure[]{\label{c=10}
\includegraphics[width=0.32\textwidth]{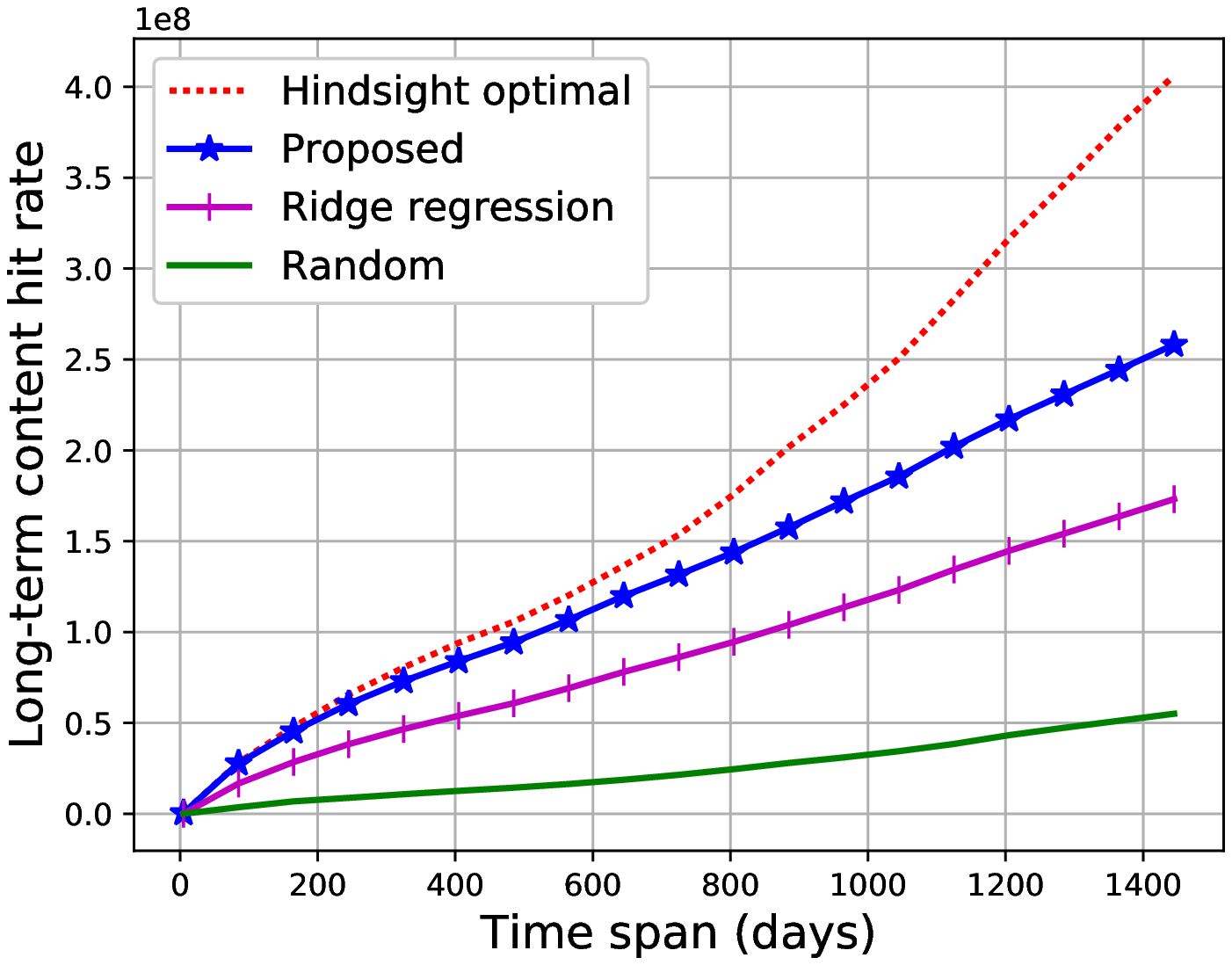}}
\subfigure[]{\label{c=40}
\includegraphics[width=0.32\textwidth]{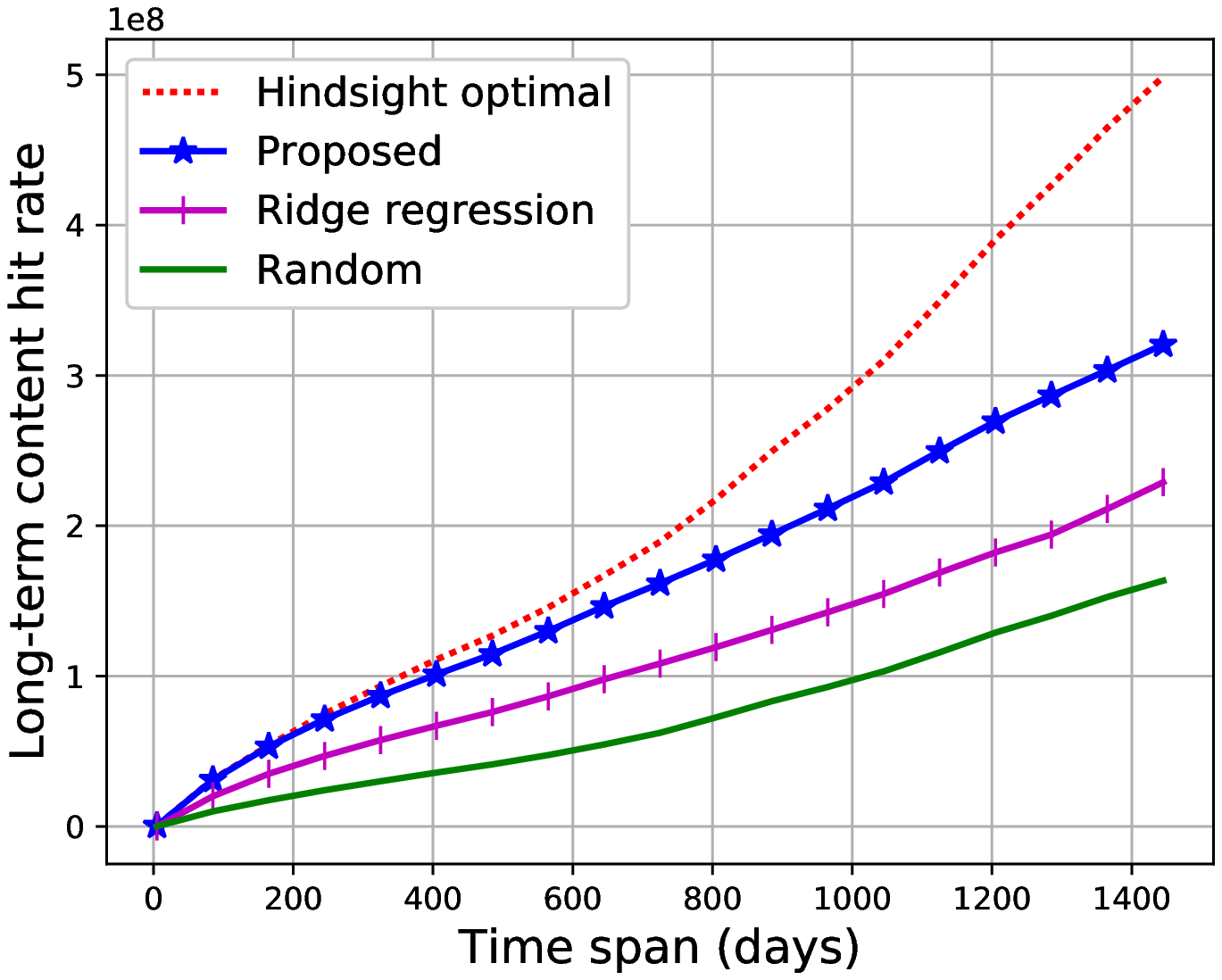}}
\subfigure[]{\label{c=80}
\includegraphics[width=0.32\textwidth]{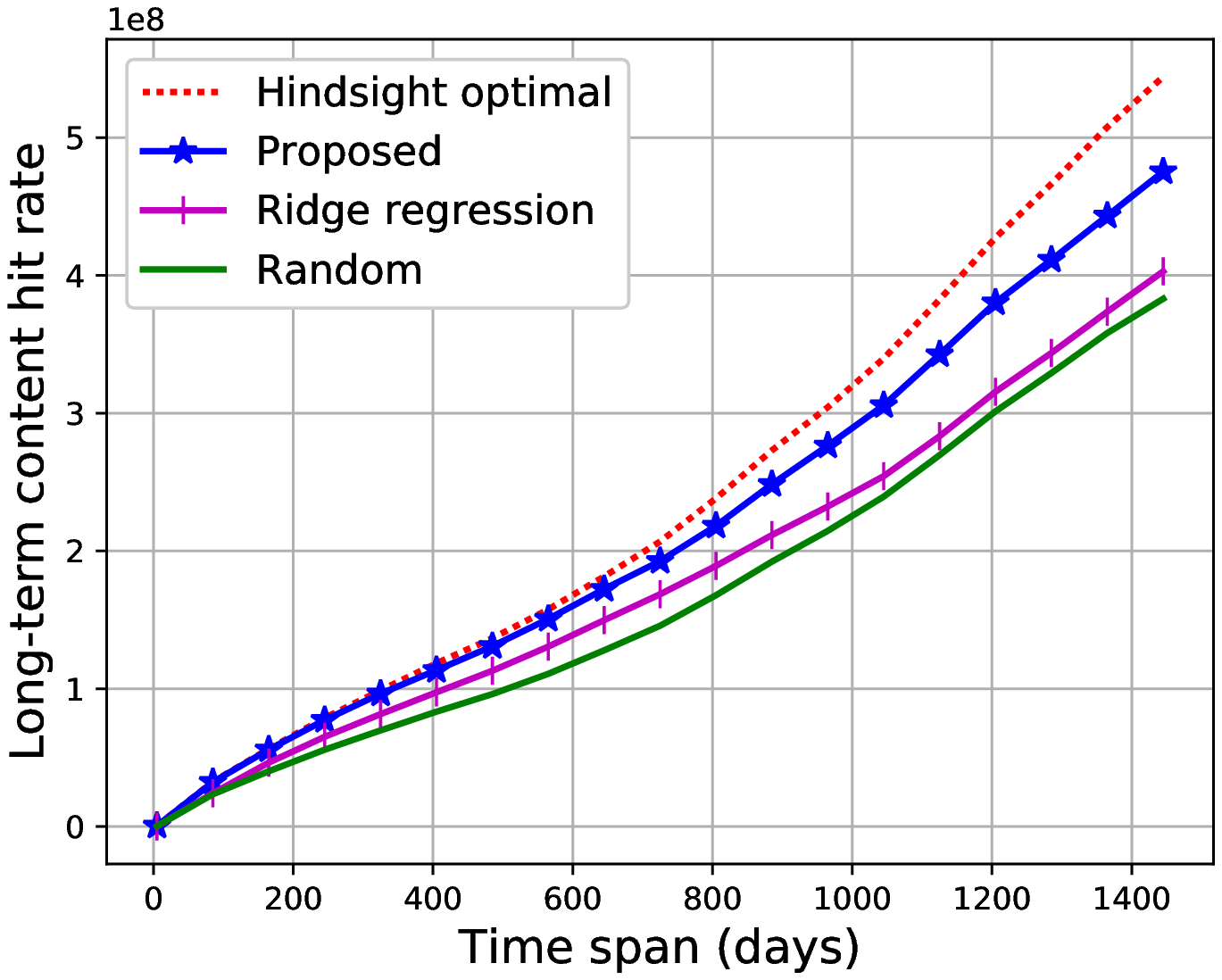}}
\caption{The content hit rate comparison between the proposed algorithm and other benchmarks with varying caching size, where the total number of videos is 100. (a) EN caching size  $c=10$, (b) EN caching size $c=30$ and c) EN caching size $c=70$.}
\label{simresult}
\end{figure*}

Based on two lemmas from \cite{Yasin} (Lemma 10 and 11), the second term in Eq. (\ref{dividedregret}) can be bounded as
\begin{eqnarray}
&&\sum_{t\in\mathcal T} \sum_{n\in\mathcal N}  r_{n,t}|_{\bar{\mathcal X}_{n,t}} \;\;\leq\;\;2\sum_{t\in\mathcal T} \sum_{n \in\mathcal N}\sum_{f \in \mathcal F_{n,t}\setminus \mathcal F_{n,t}^\ast} p_{f,n,t}\nonumber\\
&\leq& 2c\alpha_T \sum_{n\in\mathcal N} \sum_{t\in\mathcal T} \sqrt{\boldsymbol{x}_{f,n,t}^\top \boldsymbol{V}_{f,n}^{-1} \boldsymbol{x}_{f,n,t}}\label{midalpha}\\
&\leq& 2c\alpha_T \sum_{n\in\mathcal N} \sqrt{T\sum_{t\in\mathcal T} \boldsymbol{x}_{f,n,t}^\top \boldsymbol{V}_{f,n}^{-1} \boldsymbol{x}_{f,n,t}}\label{midroot}\\
&\leq& 2c\alpha_TN \sqrt{2T \ln\frac{(\lambda + T\eta^2/d)^d}{\lambda}},\label{midlemmas}
\end{eqnarray}
where Eq. (\ref{midalpha}) is due the fact that $\alpha_t$ increases with $t$, Eq. (\ref{midroot}) holds because the arithmetic mean of a set of values is smaller than their root-mean square and Eq. (\ref{midlemmas}) is based on the lemmas from \cite{Yasin}. By substituting Eq. (\ref{midlemmas}) and (\ref{firstbound}) into Eq. (\ref{dividedregret}), and together with $\alpha_T = \sqrt{\ln (TF^{\frac{1}{2}})} + \zeta\lambda$, we have 
\begin{eqnarray}
R(T) &\leq& 2c\alpha_TN\sqrt{2T\ln\frac{(\lambda + T \eta^2/d)^d}{\lambda}} + \frac{\pi^2}{3}c\gamma N \nonumber\\
& = & O\Big(cN\sqrt{dT (\ln T) \ln(\lambda + T\eta^2/d)}\Big),
\end{eqnarray}
which concludes the proof.
\end{proof}
Basically, the regret consists of two parts: estimation error and perturbation. In particular, the estimation error consists of the ridge regression model error and the intended bias incur with $\lambda>0$. The perturbation term is well managed by the varying control parameter $\alpha_t$. Theorem \ref{theorem} indicates that the proposed algorithm achieves sublinear long-term regret, i.e., $\frac{R(T)}{T}\to 0$ when $T\to\infty$, which means that the accumulative content hit rate asymptotically approaches the optimal caching scheme in the long term.

\section{Performance Evaluation}\label{simulation}
To evaluate the proposed algorithm, we conduct a study on a dataset crawled from YouTube. On YouTube, some video owners made their video view statistics open to public. Among others, the view amount information is recorded on a daily basis. We randomly crawled $100$ videos that were uploaded before January of 2013, with full view history till January of 2017. The most popular video has been watched tens of thousands of times everyday since it was uploaded, while the least popular one has been rarely viewed across the time span.

\subsection{Simulation Setup}
Note that YouTube videos are being watched globally, and we do not have access to the statistics of videos at different locations. To emulate the video view process in different places, we shift the statistics of each video backward and forward on the time span. In this way, we are able to characterize the location-related features based on different view statistics and meanwhile, maintain the temporal feature of each video record. Specifically, we consider the content library containing those $100$ videos. Each video can be cached on $3$ ENs, each with caching size $c$. Content refreshing is performed upon the network status. For example, network traffic presents regular peak and valley every day. Hence, content refreshing can be performed during the off-peak period with minimized impact on the normal network activity. Meanwhile, we use the view amount in the past $5$ days as the feature vector, i.e., $d = 5$.

We compare the proposed algorithm with the following benchmarks. 1) Hindsight optimal. Based on the full view record across the time span, the most popular videos are always selected and cached. Note that this benchmark requires future information and cannot be implemented in practice. 2) Ridge regression. As a degraded version of our proposed algorithm, the ridge regression does not account for the random noise of user demand. 3) Random. A random set of videos is selected to update the EN cache at each time slot.

\subsection{Simulation Result}

Figure \ref{simresult} shows the results of different algorithms in terms of long-term content hit rate with different EN caching sizes. It can be seen that the proposed algorithm outperforms other schemes under all caching schemes. This is because our algorithm chooses to be optimism in face of uncertainty, which helps to better identify the popular contents even under varying popularity profile. As it was found in \cite{imc}, YouTube video requests are highly skewed, indicating that a small portion of popular contents are attracting the majority of requests. This is confirmed by Fig. \ref{simresult} since the content hit rate does not grow linearly with the cache size. The performance gain of our algorithm can be higher if content popularity profile is less skewed, since contents with highly skewed popular profile can also be easily identified by other algorithms. Also note that, hit rate of the optimum stays almost the same from $c=30$ to $c=70$. This is due to the fact that content popularity is long-tailed \cite{imc}, namely, the less popular contents attract almost vanishing requests compared to the popular ones.

\section{Conclusion}\label{conclusion}
This paper proposes an online learning algorithm for dynamic mobile edge caching, by exploiting location related features. The algorithm first estimates the hit rate of a content at a specific location based on a linear model. Noticing that the accuracy may be affected by random noises, a perturbation is added to the estimation to account for uncertainty. Then, according to estimation results, contents that are predicted to maximize hit rate at a certain location are cached respectively. Theoretical analysis indicates that the proposed algorithm achieves sublinear long-term regret when compared to the optimal caching policy. Simulations on real world traces demonstrate the advantage of the proposed algorithm. For future work, we will investigate the impact of location differentiation on coded caching schemes.

\appendix[Proof of Lemma \ref{estimationerror}]
Let $\boldsymbol{h}_{f,n} = \boldsymbol{\Phi}_{f,n}^\top \boldsymbol{y}_{f,n}$, based on Eq. (\ref{ridge}), the estimation error can be rewritten as
\begin{eqnarray}
&&|\boldsymbol{x}_{f,n}^\top \tilde{\boldsymbol{\theta}}_n - \boldsymbol{x}_{f,n}^\top \boldsymbol{\theta}_n^\ast| \nonumber\\
&=& |\boldsymbol{x}_{f,n}^\top \boldsymbol{V}_{f,n}^{-1} \boldsymbol{h}_{f,n} - \boldsymbol{x}_{f,n}^\top \boldsymbol{V}_{f,n}^{-1} (\boldsymbol{\Phi}_{f,n}^\top \boldsymbol{\Phi}_{f,n} + \lambda \boldsymbol{I}_d) \boldsymbol{\theta}_n^\ast| \nonumber \\
&=& |\boldsymbol{x}_{f,n}^\top \boldsymbol{V}_{f,n}^{-1} \boldsymbol{\Phi}_{f,n}^\top (\boldsymbol{y}_{f,n} - \boldsymbol{\Phi}_{f,n}\boldsymbol{\theta}_n^\ast) - \lambda \boldsymbol{x}_{f,n}^\top \boldsymbol{V}_{f,n}^{-1} \boldsymbol{\theta}_n^\ast|. \nonumber 
\end{eqnarray}
Since $||\boldsymbol{\theta}_n^\ast|| \leq \zeta$, according to H\"older's inequality,
\begin{eqnarray}
|\boldsymbol{x}_{f,n}^\top \tilde{\boldsymbol{\theta}}_n - \boldsymbol{x}_{f,n}^\top \boldsymbol{\theta}_n^\ast| & \!\! \leq \!\! & |\boldsymbol{x}_{f,n}^\top \boldsymbol{V}_{f,n}^{-1} \boldsymbol{\Phi}_{f,n}^\top (\boldsymbol{y}_{f,n} - \boldsymbol{\Phi}_{f,n}\boldsymbol{\theta}_n^\ast)|\nonumber\\
&& + \;\; \zeta\lambda ||\boldsymbol{x}_{f,n}^\top \boldsymbol{V}_{f,n}^{-1}|| \label{twoerror}. 
\end{eqnarray}
The right-hand side of above inequality decomposes the estimation error into two parts, with the first (variance term) specifies the error caused by linear model, and the second (bias term) is the bias incurred by ridge regression parameter $\lambda$. According to Eq. (\ref{linearprediction}), we have $\mathbb E [\boldsymbol{y}_{f,n} - \boldsymbol{\Phi}_{f,n}\boldsymbol{\theta}_n^\ast] = 0$. The Azuma's inequality gives an probabilistic upper bound of the variance term of Eq. (\ref{twoerror}):
\begin{flalign}
&\mathbb{P}\bigg\{|\boldsymbol{x}_{f,n}^\top \boldsymbol{V}_{f,n}^{-1} \boldsymbol{\Phi}_{f,n}^\top (\boldsymbol{y}_{f,n} - \boldsymbol{\Phi}_{f,n}\boldsymbol{\theta}_n^\ast)| > \delta \sqrt{\boldsymbol{x}_{f,n}^\top \boldsymbol{V}_{f,n}^{-1} \boldsymbol{x}_{f,n}}\bigg\}&& \nonumber \\
&\leq  2\exp\Big(-\frac{2\delta^2 \boldsymbol{x}_{f,n}^\top \boldsymbol{V}_{f,n}^{-1} \boldsymbol{x}_{f,n}} {||\boldsymbol{x}_{f,n}^\top \boldsymbol{V}_{f,n}^{-1} \boldsymbol{\Phi}_{f,n}^\top||^2}\Big) \leq 2e^{-2\delta^2},&& \label{error1}
\end{flalign}
where the last inequality is due to the fact that 
\begin{eqnarray}
\boldsymbol{x}_{f,n}^\top \boldsymbol{V}_{f,n}^{-1} \boldsymbol{x}_{f,n} &=& \boldsymbol{x}_{f,n}^\top \boldsymbol{V}_{f,n}^{-1} (\boldsymbol{\Phi}_{f,n}^\top \boldsymbol{\Phi}_{f,n} + \lambda \boldsymbol{I}_d) \boldsymbol{V}_{f,n}^{-1} \boldsymbol{x}_{f,n} \nonumber \\
& \geq & \boldsymbol{x}_{f,n}^\top \boldsymbol{V}_{f,n}^{-1} \boldsymbol{\Phi}_{f,n}^\top \boldsymbol{\Phi}_{f,n} \boldsymbol{V}_{f,n}^{-1} \boldsymbol{x}_{f,n} \nonumber \\
& = & ||\boldsymbol{x}_{f,n}^\top \boldsymbol{V}_{f,n}^{-1} \boldsymbol{\Phi}_{f,n}^\top||^2.
\end{eqnarray}
Hence, the variance term of Eq. (\ref{twoerror}) can be bounded by $\delta \sqrt{\boldsymbol{x}_{f,n}^\top \boldsymbol{V}_{f,n}^{-1} \boldsymbol{x}_{f,n}}$ with probability at least $1-2e^{-2\delta^2}$. Further, The bias term of Eq. (\ref{twoerror}) can be bounded as 
\begin{eqnarray}
||\boldsymbol{x}_{f,n}^\top \boldsymbol{V}_{f,n}^{-1}|| & = & \sqrt{\boldsymbol{x}_{f,n}^\top \boldsymbol{V}_{f,n}^{-1} \boldsymbol{I}_d \boldsymbol{V}_{f,n}^{-1} \boldsymbol{x}_{f,n}} \nonumber \\
& \leq & \sqrt{\boldsymbol{x}_{f,n}^\top \boldsymbol{V}_{f,n}^{-1} (\lambda\boldsymbol{I}_d+ \boldsymbol{\Phi}_{f,n}^\top \boldsymbol{\Phi}_{f,n}) \boldsymbol{V}_{f,n}^{-1} \boldsymbol{x}_{f,n}} \nonumber \\ 
& = & \sqrt{\boldsymbol{x}_{f,n}^\top \boldsymbol{V}_{f,n}^{-1} \boldsymbol{x}_{f,n}}. \label{error2}
\end{eqnarray}
By substituting Eq. (\ref{error1}) and (\ref{error2}) into Eq. (\ref{twoerror}), the probabilistic bound in Eq. (\ref{errorbound}) directly follows.

\section*{Acknowledgment}
This work is supported by National Natural Science Foundation of China under Grant No. 61231010, Research Fund for the Doctoral Program of MOE of China under Grant No. 20120142110015 and the Natural Sciences and Engineering Research Council (NSERC) of Canada. Peng Yang is also financially supported by the China Scholarship Council.

\end{document}